\documentclass[a4paper,12pt]{article}
\usepackage{amssymb}
\usepackage{amsmath}
\usepackage{amsfonts}
\usepackage{psfig}
\usepackage{colortbl}
\usepackage{amsthm}
\usepackage{cite}

\newtheorem{theorem}{Theorem}
\newtheorem{proposition}{Proposition}
\newtheorem{lemma}{Lemma}
\newtheorem{definition}{Definition}
\newtheorem{remark}{Remark}

\def\bra#1{\left\langle #1\right|}
\def\ket#1{\left| #1\right\rangle}
\def\bracket#1#2{\left\langle #1 | #2 \right\rangle}

\def\half{\frac{1}{2}}
\def\del#1{\partial_{#1}}
\def\nn{\nonumber}
\def\beq{\begin{equation}}
\def\eeq{\end{equation}}
\def\bea{\begin{eqnarray}}
\def\eea{\end{eqnarray}}
\def\SDsum{\supset \hspace{-1em}\hspace{-1pt}+}
\def\alg{{\mathfrak g}}
\def\calg{\tilde{\mathfrak g}_{\ell}}
\def\Rep{ {\cal V}^{a,b}_{AB} }

\renewcommand{\arraystretch}{1.3}

\begin{document}

\begin{center}
 {\huge Galilean conformal algebras in two spatial dimension}

\vspace{1cm}

Naruhiko Aizawa and Yuta Kimura

\bigskip
Department of Mathematics and Information Sciences, \\ 
Graduate School of Science, Osaka Prefecture University, \\
Nakamozu Campus, Sakai, Osaka 599-8531, Japan.

\end{center}
\begin{abstract}

 A class of infinite dimensional Galilean conformal algebra  in 
$ (2+1) $ dimensional spacetime is studied. Each member of the class, 
denoted by $ \alg_{\ell}, $ is labelled by the parameter $ \ell. $ 
The parameter $ \ell $ takes a spin value, i.e., 1/2, 1, 3/2, \dots.  
We give a classification of all possible central extensions of $ \alg_{\ell}. $ 
Then we consider the highest weight Verma modules over $ \alg_{\ell}$ with the central extensions.  
For integer $ \ell $ we give an explicit formula of Kac determinant. 
It results immediately that  the Verma modules are irreducible for nonvanishing highest weights. 
It is also shown that the Verma modules are reducible for vanishing highest weights. 
For half-integer $ \ell $ it is shown that all the Verma module is reducible. 
These results are independent of the central charges.
\\


\end{abstract}

\section{Introduction}
\label{Sec:Intro}

 The study of possible kinematical invariance lager than Galilei group in nonrelativistic physics  
introduced the notion of Schr\"odinger group~\cite{Nie,Hagen}. 
The group contains dilatation and nonrelativistic conformal transformation in addition to 
those of Galilei group. 
It was followed by further enlargement of the group by geometric consideration~\cite{HavPre,NdORM1,NdORM2} 
so that we have at present some distinct kinematic groups for nonrelativistic systems which 
contains the Schr\"odinger group as a special case. 
They are regarded as nonrelativistic analogue of the conformal groups $SO(n,2) $ and referred to 
as Galilean conformal groups (also called conformal Galilei groups) and conformal Newton-Hooke groups. 

  Those nonrelativistic conformal groups  and their Lie algebras are non-semisimple and finite dimensional. 
Recent observation of nonrelativistic AdS/CFT correspondence caused a renewed interest in 
those algebraic structures. It was observed first for the Schr\"odinger algebra~\cite{Son,BaMcG} and then 
for the Galilean conformal algebras~\cite{BaGop,AlDaVa,MT}. 
Other areas of physics where the nonrelativistic conformal groups play certain roles 
range from classical mechanics to quantum field theory (see for example the references in \cite{AiIsKi,RoUn}). 
Among others, the Galilean conformal algebras in $(2+1)$ dimensional spacetime are of particular 
interest, since they have the socalled \textit{exotic} central extension. 
That is the central extension existing only in this particular dimension of spacetime and 
has different structure from those in other dimensional spacetime~\cite{MT,StZak,LSZ,LSZ2}. 
We remark that the exotic central extension also exist for the (non conformal) Galilei algebra 
and has been studied extensively. Further detail of this, see for example~\cite{Horv} and 
references therein. 

  Another interesting aspect of the Galilean conformal algebras is extensions to infinite 
dimensional Lie algebras. 
Several different infinite dimensional extensions have been introduced so far~\cite{BaGop,AlDaVa,MT,Henkel,RoUn,BGMM,HR,HR2,HSSU}.  
The Virasoro algebra, one of the most important algebra is mathematical physics, is a subalgebra of all 
those extensions. Some of the extensions contains two or more copies of the Virasoro algebras. 
Probably the best known example of such extensions is the  Schr\"odinger-Virasoro algebra~\cite{Henkel,RoUn}. 
The algebra is applied to some statistical systems~\cite{Henkel2,HeUn} and partial differential equations~\cite{CheHen}. 
Representation theory of the Schr\"odinger-Virasoro algebra has been studied extensively~\cite{RoUn,SchRep1,SchRep2,SchRep3,SchRep4}. 
Other varieties of infinite dimensional Galilean conformal algebras appear in various contexts such as 
topologically massive gravity, nonlinear partial differential equations, Navier-Stokes equations 
and so on~\cite{DuHoNewCar,HoKuNi,BagchiTop,Mukh,CheHen2,BagchiTension,Hosseiny,BaBaMe}. 
The same algebraic structure is discussed not in the nonrelativistic conformal symmetry but in completely another contexts. 
As the isometry of flat Minkowski space at null infinity, 
the algebra is called BMS (Mondi-Metzner-Sachs) algebra~\cite{BagchiBMS,BagFare,BarGomGon}. 
In a relation to vertex operator algebras it is called $ W(a,b)$-algebra~\cite{ZhangDong,LiuGaoZhu,GaoJiangPei2,Radobolja}. 

   In the present work we investigate the infinite dimensional Galilean conformal algebras introduced by 
Martelli and Tachikawa~\cite{MT}. We shall focus on the algebras defined in $(2+1)$ dimensional spacetime 
because of the following reasons: i) as already mentioned, finite dimensional counterparts have a particular central 
extension only in this dimension of spacetime. ii) there are several publications discussing the representation theory and 
physical applications of the same algebras defined in $(1+1)$ dimensional spacetime~\cite{BGMM,HR,HSSU,HoKuNi,BagchiTop,Hosseiny,BagchiBMS,BagFare,BarGomGon,ZhangDong,LiuGaoZhu,GaoJiangPei2,Radobolja}. 
However, the case of $(2+1)$ dimensional counterparts are not studied yet. 
Our main results are a classification of all possible central extensions and 
a criterion for irreducibility of Verma modules. 
The latter is a consequence of our explicit formula of Kac determinant which is also our main result. 
Some of the preliminary results (for the simplest member of the algebras defined in $(2+1)$ dimension) 
have already been reported elsewhere~\cite{Naru1,Naru2}. 

This article is organized as follows: 
In the next section we give a definition of the algebras. Each algebra is labelled by 
a positive integer or a positive half-integer. 
We try to classify central extensions of the algebras in \S \ref{sec:extension}. 
It will be shown that the exotic central extension is not allowed for 
the infinite dimensional algebras. 
In \S \ref{sec:verma} we define Verma modules over the algebra labelled by a positive integer, then study their reducibility. 
An explicit formula of Kac determinant is given. In \S \ref{sec:Verma2} Verma modules over the algebra 
labelled by a positive half-integer is defined and it will be shown that all the Verma modules is reducible.

%
%
%
\setcounter{equation}{0}
\section{Definition and structure}

  We employ the definition of the infinite dimensional Galilean conformal algebras (GCA) introduced in~\cite{MT}. 
The algebras in \cite{MT} are a natural extension of the finite dimensional counterparts defined in~\cite{NdORM1,NdORM2}.   
These algebras are labelled by two parameters $(d,\ell) $  where $d$ is interpreted as a dimension of 
space on which the GCA acts as an infinitesimal transformation. 
The parameter $ \ell, $ sometimes called ``spin", takes a positive integer or a positive 
half-integer value. 
The smallest instance $\ell = \half $ corresponds to the twisted Schr\"odinger-Virasoro algebra introduced in~\cite{Henkel,RoUn}. 
The mathematical and physical aspects of the algebra labelled by $(d,\ell)=(1,1) $ has been investigated in literatures~\cite{BGMM,HR,HSSU,HoKuNi,BagchiTop,BagchiBMS,BagFare,BarGomGon,ZhangDong,LiuGaoZhu,GaoJiangPei2,Radobolja} 
(see also \cite{Hosseiny} for $d = 1 $ and any $\ell$). 
In the present work we restrict ourselves to the algebras with $ d = 2. $   
This class of algebra has generators of three different types: 
$ L_m, \; J_m, \; P_r^i $ with $ m \in {\mathbb Z},\; r \in {\mathbb Z} + \ell $ and $ i,j = 1, 2.$ 
They are subject to the relations:
\bea
  & & [L_m, L_n] = (m-n) L_{m+n}, \qquad [J_m, J_n] =   
      [P_r^i, P_s^j] = 0, 
    \nn \\
  & & [L_m, J_n] = -n J_{m+n}, \qquad\qquad\; [L_m, P_r^i] = (\ell m-r) P_{m+r}^i, 
   \nn \\
  & & [J_m, P_r^i] = \sum_j \epsilon_{ij} P_{m+r}^j,
      \label{defrel}
\eea
where $ \epsilon_{ij} $ is the antisymmetric tensor with $  \epsilon_{12} =  1. $ 
The $ \langle L_m \rangle $ sector is a centerless Virasoro subalgebra and the $ \langle P_m^i \rangle $ 
sector is an Abelian ideal. The structure of the algebra is as follows:
\[
  (\; \langle L_m \rangle \; \SDsum \; \langle J_m \rangle \; ) \; 
  \SDsum \; \langle P_m^i \rangle, 
\]
where $ \SDsum $ denotes the semidirect sum. 
We denote the algebra for a fixed value of $ \ell$  by $\alg_{\ell}.$ 
The algebra $ \alg_{\ell} $ is realized in terms of the space-time coordinates $ (t, x_1, x_2)$ \cite{MT}:
\bea
  & & L_m = -t^{m+1}\del{t} - \ell (m+1) t^m \sum_i x_i \del{i}, 
    \nn \\
  & & P_r^i = -t^{r+\ell} \del{i}, \qquad
      J_m = -t^m ( x_1 \del{2} - x_2 \del{1}). \label{g-real}
\eea
There exists an algebraic anti-automorphism $ \omega : \alg_{\ell} \to \alg_{\ell} $ given by
\beq
   \omega(L_m) = L_{-m}, \qquad \omega(J_m) = -J_{-m}, \qquad \omega(P_r^i) = P_{-r}^i.
   \label{antiauto}
\eeq
It will be used later in the consideration of highest weight representations. 

  The subalgebra 
$ \langle \, L_0, L_{\pm 1}, J_0, P^i_{-\ell}, P^i_{-\ell+1}, \dots, P^i_{\ell} \, \rangle $ 
is isomorphic to the finite dimensional Galilean conformal algebras of \cite{NdORM1,NdORM2}. 
It is known that the finite dimensional Galilean conformal algebras have different central extensions 
depending on the parity of $ 2\ell $ \cite{MT,StZak,LSZ}. Especially, the central extension for 
integer $ \ell $ (called the exotic extension) exists only for $ d = 2. $ All the central extensions make 
the commuting subalgebra $ \langle P_r^i \rangle $ noncommutative. 
On the other hand, central extensions of the $d=2$ infinite dimensional algebra $\alg_{\ell}$ 
have not studied yet. 

%
%
%
\setcounter{equation}{0}
\section{Central extensions of $\alg_{\ell}$}
\label{sec:extension}

It may be natural to ask whether the infinite dimensional algebra $ \alg_{\ell} $ also 
has the central extensions. To answer the question, we try to classify all 
possible central extensions of $\alg_{\ell}. $ 
Classification of the central extensions for the $ d=1 $ infinite dimensional Galilean conformal algebras 
has been done in \cite{Hosseiny}. 
The following theorem is our main result in this section:
\begin{theorem}
\label{theo:ext}
  All possible central extensions of $ \alg_{\ell} $ are listed as follows:
  \bea
   & & [L_m, L_n] = (m-n) L_{m+n} + \frac{c_1}{12} m(m^2-1) \delta_{m+n,0},
    \nn \\
   & & [J_m, J_n] = c_2 m \delta_{m+n,0}, \nn \\
   & & [L_m, J_n] = -nJ_{m+n} + c_3 m^2 \delta_{m+n,0}, \nn
  \eea
  where $ c_1, c_2 $ and $ c_3 $ are independent central charges. 
\end{theorem}
Contrary to the finite dimensional Galilean conformal algebras, 
the infinite dimensional counterparts do not have the central extensions which make 
the abelian subalgebra $ \langle P_r^i \rangle $ noncommutative. 
The situation similar to this is also observed in the Galilean line group \cite{MacMcWick}. 

\begin{proof}
We take a pedestrian way to prove the theorem. 
We add the central terms to each commutators in (\ref{defrel}). 
\bea
  & & [L_m, L_n] = (m-n) L_{m+n} + Z_{mn}^{(L)}, \qquad [J_m, J_n] = Z_{mn}^{(J)}, \qquad  
      [P_r^i, P_s^j] = Y_{rs}^{ij}, 
   \nn \\
  & & [L_m, J_n] = -n J_{m+n} + C_{mn}, \qquad\qquad\; [L_m, P_r^i] = (\ell m-r) P_{m+r}^i + F_{mr}^{\ i}, 
   \nn\\
  & & [J_m, P_r^i] = \sum_j \epsilon_{ij} P_{m+r}^j + W_{mr}^{\ i},  \label{D2L1IC} 
\eea
where $  Y_{rs}^{ij} = -Y_{sr}^{ji} $ are required from the antisymmetry of the bracket. 
Other central terms are antisymmetric with respect to $ m, n, r $ and $s.$ 
We impose the Jacobi identities on the extended commutation relations (\ref{D2L1IC}).  
The Jacobi identity for $ \{ L_k, L_m, L_n \} $ yields the relation:
\beq
 (k-m) Z_{k+m\; n}^{(L)} + (m-n) Z_{m+n\; k}^{(L)} + (n-k) Z_{n+k\; m}^{(L)} = 0. 
 \label{ZLrel}
\eeq
The following relations are obtained from the Jacobi identities for $ \{ L_m, L_k, L_n \} $ and 
$ \{ L_m, J_k, J_n \}:$ 
\bea
  & & (m-k) C_{m+k\;n} + n C_{m\;k+n} - n C_{k\;n+m} = 0, 
    \label{Crel} \\
  & & k Z_{k+m\; n}^{(J)} - n Z_{n+m\; k}^{(J)} = 0.
    \label{ZJrel}
\eea
For $ \{ L_m, P_r^i, P_s^j \} $ and $ \{ J_m, P_r^i, P_s^j \} $ we have the relations:
\bea
   & & (\ell m-r) Y_{m+r\;s}^{ij} - (\ell m-s) Y_{m+s\;r}^{ji} = 0, 
    \label{Yrel} \\
   & & Y_{m+r\;s}^{11} + Y_{m+s\;r}^{22} = 0, \qquad 
   Y_{m+r\;s}^{ij} - Y_{m+s\;r}^{ij} = 0, \quad i \neq j, 
     \label{Yrel3}
\eea
Finally for $ \{ L_m, L_n, P_r^i \}, \; \{ L_m, J_n, P_r^i \} $ and $ \{ J_m, J_n, P_r^i \} $ 
the following relations are obtained:
\bea
  & & (m-n) F_{m+n\;r}^{\ i} - (\ell n-r) F_{m\; n+r}^{\ i} + (\ell m-r) F_{n\; m+r}^{\ i} = 0,
     \label{Frel} \\
  & &  F_{m\;n+r}^{\ i} = \sum_j \epsilon_{ij} ( n W_{m+n\;r}^{\ j} - (\ell m-r) W_{n\;m+r}^{\ j}).
       \label{FWrel1} \\
  & & W_{m\;n+r}^{\ i} - W_{n\;m+r}^{\ i} = 0.
     \label{Wrel}
\eea
No other relations are deduced from the Jacobi identities. 
It is observed from (\ref{ZLrel}) that the central element $ Z_{mn}^{(L)} $ for the Virasoro subalgebra decouples from others. 
It follows that the well-known central extension of the Virasoro algebra remains true for $\alg_{\ell}.$ 

  The central terms $ Z_{mn}^{(J)} $ and $ C_{mn} $ are also decouple from others. 
We show that the relations (\ref{Crel}) and (\ref{ZJrel}) give nontrivial extensions. 
Set $ k = 0 $ in (\ref{Crel}). We then have $ C_{mn} = n(m+n)^{-1} C_{0\;m+n} $ if $ n+m \neq 0. $ 
This extension is absorbed in $ J_m $ by the redefinition $ J_m' =  J_m - m^{-1} C_{0m} $ 
so that this is trivial. The only possibility for $ C_{mn} $ is $ C_{mn} = \delta_{m+n,0} \, g(m) $ with 
$ g(m) = -g(-m). $  It follows immediately that $ g(0) = 0. $  
Substitution of this into (\ref{Crel}) yields the relation
\beq
  (m-k)g(m+k) + (m+k) (g(k)-g(m)) = 0. \label{Crel2}
\eeq
Setting $ k = 1 $ in (\ref{Crel2}) one has the linear recurrence relation for $ g(m):$
\beq
  (m-1) g(m+1) - (m+1) (g(m)-g(1)) = 0. \label{Crel3}
\eeq
The solution space of  (\ref{Crel3}) is at most two dimensional since one may obtain $g(m)$ if one knows 
$ g(1) $ and $ g(2). $ It is easy to see that $ g(m) = m, m^2 $ are the two independent solutions. 
The general solution to (\ref{Crel3}) is $ g(m) = \mu m + \nu m^2. $ We set $ \mu = 0 $ in this paper. 

  Now we turn to $ Z_{mn}^{(J)}. $ Set  $ m=0 $ in  (\ref{ZJrel}) then we have the identity
\[
  (k+n) Z^{(J)}_{kn} = 0,
\]
which implies that $ Z^{(J)}_{kn} = \delta_{k+n,0} f(k) $ with $ f(k) = -f(-k). $ 
With this form of $ Z^{(J)}_{kn} $ the relation (\ref{ZJrel}) requires the identity for $ f(k,k):$ 
\[
   kf(n) -n f(k) = 0.
\]
This is solved by $ f(n) = const \times n. $ Thus $ Z^{(J)}_{kn} $ gives a nontrivial extension.

  Next we show that the relations from (\ref{Yrel}) to (\ref{Wrel}) give no nontrivial extensions. 
We start with showing that  $ Y_{rs}^{ii} = 0. $ Set $ j=i $ and $ m=0 $ in (\ref{Yrel}), then we have 
\[
   (r+s) Y_{rs}^{ii} = 0.
\]
It follows that $ Y_{rs}^{ii} = \delta_{r+s,0} f^i(r) $ with $ f^i(r) = -f^i(-r). $ 
Substitution this into (\ref{Yrel}) yields 
\beq
  (\ell m- r) f^i(m+r) + ( (\ell+1) m + r ) f^i(r) = 0. \label{Yrel2}
\eeq
We treat the cases of integer $ \ell $ and half-integer $ \ell $ separately. 

\noindent
(i) if $ \ell $ is an integer, then so is $r$ and $ f^i(0) = 0. $ 
Set $ r=0 $ in (\ref{Yrel2}) we have 
\[
  \ell f^i(m) + (\ell+1) f^i(0) = 0, 
\]
for any nonvanishing $ m. $ This means that $ f^i(m) = 0 $ for all $m.$ 

\noindent
(ii) if $ \ell $ is a half-integer, then so is $r. $ Set $ m = 1$ in (\ref{Yrel2}) we have 
\beq
  (\ell-r) f^i(r+1) + (\ell+1+r) f^i(r) = 0. \label{Yrel4}
\eeq
This relates any $ f^i(r) $ with $ f^i(\ell). $ By setting $ r = \ell $ in (\ref{Yrel4}) 
we see that $ f^i(\ell) = 0. $ Thus $ f^i(r) = 0 $ for any half-integer $ \ell. $ 

\noindent 
This completes the proof of $ Y_{rs}^{ii} = 0. $ 
We now consider $ Y_{rs}^{ij} $ with $ i \neq j. $ 
Setting $ m = 0 $ in (\ref{Yrel3}) one sees that $ Y_{rs}^{ij} = - Y^{ji}_{rs}. $ 
This implies that $ Y_{rs}^{ij} $ has the structure 
$ Y_{rs}^{ij} = \epsilon_{ij} y(r,s) $ with a symmetric $y(r,s). $ 
It then follows from (\ref{Yrel}) that 
\beq
  (\ell m-r) y(m+r,s) + (\ell m-s) y(m+s,r) = 0.  \label{Yrel5}
\eeq
Set $ m = 0 $ in this equation. Then 
\[
  (r+s) y(r,s) = 0,
\]
so that $ y(r,s) = \delta_{r+s,0} h(r) $ with $ h(r) = h(-r). $ 
Put this form into (\ref{Yrel}) and (\ref{Yrel5}) we have the relations:
\bea
  & & h(m+r) = h(r), \nn \\
  & & (\ell m-r) h(m+r) + ((\ell+1)m+r) h(r) = 0.
\eea
The first equation implies that $ h(r) $ is a constant. From the second equation 
one may see that the constant is zero. This completes the proof of $ Y_{rs}^{ij} = 0. $ 

 Finally we show that equations (\ref{Frel}), (\ref{FWrel1}) and (\ref{Wrel}) do not produce any 
nontrivial extensions. 
Set $ m = 0 $ in (\ref{Frel}) and (\ref{FWrel1}):
\bea
  & & (n+r) F_{nr}^{\; i} = -(\ell n-r) F_{0\,n+r}^{\; i}, \nn \\
  & & F_{0\,n+r}^{\; i} = (n+r) \sum_j \epsilon_{ij}  W_{nr}^{\; j}. \nn 
\eea
If $ n +r \neq 0, $ the first equation gives the central extension
\beq
  F_{nr}^{\; i} = -\frac{\ell n- r}{n+r} F_{0\,n+r}^{\; i}. \label{F-extension-trivial}
\eeq
Together with the second equation, one can see without any difficulty that the extension (\ref{F-extension-trivial}) 
is absorbed by the redefinition 
\[
  P_r^i \ \to \ P_r^i - r^{-1} F_{0r}^{\; i}. 
\]
Thus the extension (\ref{F-extension-trivial}) is trivial and the only possibility of the nontrivial extension 
is $ n + r = 0. $ This relation is never true if $ \ell $ is a half-integer. Thus we have shown that 
there is no nontrivial $ F_{mr}^{\; i} $ if $ \ell $ is a half-integer. 
Now suppose that $ \ell $ is an integer and write $m$ instead of $r.$ Then
\[
  F_{mn}^{\; i} = \delta_{m+n,0} \varphi_m^i, \quad \varphi_m^i = - \varphi_{-m}^i
\]
The equation (\ref{FWrel1}) reads
\[
 \delta_{m+n+k,0} \varphi_m^i = \sum_j \epsilon_{ij} ( n W_{m+n\;k}^{\ j} - (\ell m-k) W_{n\;m+k}^{\ j}). 
\]
Set $ n = m+k = 0, $ then we have
\[
  \varphi_m^i = -\sum_j \epsilon_{ij} (\ell+1) m W_{00}^{\; j},
\]
since $ W_{00}^j = 0. $ This shows that there exists no nontrivial $ F_{mn}^{\; i} $ for integer $ \ell. $ 

 Non-existence of nontrivial $ F_{mr}^{\; i} $ allows us to set $ F_{mr}^{\; i} = 0 $ on the left hand side of (\ref{FWrel1}). 
Thus we have
\[
  n W_{m+n\,r}^{\; i} - (\ell m-r) W_{n\,m+r}^{\; i} = 0.
\]
Setting $m=0$ this relation yields
\[
  (n+r) W_{nr}^{\; i} = 0, 
\]
which implies that $ W_{nr}^{\; i} = 0 $ if $ n+r \neq 0. $ 
It follows that $ W_{nr}^{\; i} = 0 $ for any half-integer $ \ell. $ 
Suppose that $ \ell $ is an integer, then 
\[
    W_{mn}^{\; i} = \delta_{m+n,0} w_m^i, \quad w_m^i = - w_{-m}^i.
\]
From (\ref{Wrel}) we have the relation $ w_m^i = w_n^i, $ namely, $ w_m^i $ is a constant equal to 
$ w_0^i = 0. $ 

This completes the proof of the theorem. 
\end{proof}

%
%
%
\setcounter{equation}{0}
\section{Verma modules over $ \calg$ for integer $ \ell$}
\label{sec:verma}

\subsection{Verma modules}

 In this section we study highest weight representations,  especially 
Verma modules, of $ \alg_{\ell} $ with the central extensions. 
The extended algebra by all the central elements given in Theorem \ref{theo:ext} is denoted by $ \tilde{\alg}_{\ell}. $ 
From now on we assume that $ \ell $ is a positive integer. 
To study the Verma modules over $ \calg $ for an integer $ \ell, $ 
we employ the procedure which is an extension of that for 
Schr\"odinger-Virasoro algebra used in~\cite{RoUn}. 
Define the degree of  $X_n \in \calg $ by $ \deg(X_n) = -n $ where $ X = L, J, P^i. $ 
This allows us to define the triangular type decomposition of $ \calg: $ 
  \bea
    \calg &=& \calg^{-} \oplus \calg^0 \oplus \calg^+ 
    \nn \\
    &=& 
    \langle\; L_{-n}, J_{-n}, P_{-n}^i \; \rangle \; \oplus \;
    \langle\; L_{0}, J_{0}, P_{0}^i \; \rangle \; \oplus \;
    \langle\; L_{n}, J_{n}, P_{n}^i \; \rangle,
    \qquad
    n \in {\mathbb Z}_+ \nn
  \eea
Let $ \ket{0} $ be the highest weight vector:
   \bea
    & &  L_n \ket{0} = J_n \ket{0} = P_n^i \ket{0} = 0, \quad n \in {\mathbb Z}_+ \nn \\
    & &  L_0 \ket{0} = h \ket{0}, \quad J_0 \ket{0} = \mu \ket{0}, \quad P_0^i \ket{0} = \rho_i \ket{0},
      \nn
   \eea
Following the usual definition of Verma modules (see e.g. \cite{Dix}), we define the Verma modules over $ \calg $ by 
\[
   V^{\cal I} = U(\calg^{-}) \ket{0},
\]
where $ {\cal I} = \{\; h, \mu, \rho_1, \rho_2, c_1, c_2, c_3 \; \}. $ 
The Verma module $ V^{\cal I} $ is a graded-modules through a natural 
extension of the degree from $ \calg $ to $ U(\calg) $ by 
$ \deg(XY) = \deg(X) + \deg(Y),$ $ \; X, Y \in U(\calg), $
\[
   V^{\cal I}
       = \bigoplus_{ n \in {\mathbb Z}_{\geq 0} } V^{\cal I}_n,
     \quad
     V^{\cal I}_n = \{ X \ket{0} \ | \ X \in U(\calg^-), \ \deg(X) = n \ \}.
\]

  One can introduce an inner product in $ V^{\cal I} $ by extending the anti-automorphism 
$ \omega $ of $ \alg_{\ell} $ defined in (\ref{antiauto}) to $ U(\calg). $ 
We define the inner product of $ X \ket{0}, Y \ket{0} \in V^{\cal I} $ by
\[
   \bra{0} \omega(X) Y \ket{0}, \qquad \bracket{0}{0} = 1. 
\]
We remark that the central charges $ c_k $ are real under $ \omega. $ 

  The basis of $ V^{\cal I}_n $ is specified by a partition of an integer.  
Let us first fix our notations and conventions. 
A partition $ A=(a_1a_2\cdots a_{\ell}) $ of a positive integer $n$ is the sequence of 
positive integers such that
\bea
  n &=& a_1 + a_2 + \cdots + a_{\ell}, \nn \\
    & & a_1 \geq a_2 \geq \cdots \geq a_{\ell} > 0.  \nn
\eea
The integers $n$ and $\ell$ are called degree and length of the partition $A,$ respectively. 
They are denoted by $ \deg A $ and $ \ell(A). $ 
For a given $n,$ the number of possible partitions is denoted by $ p(n). $ 
Let $ A = (a_1a_2 \cdots a_{\ell}), B = (b_1 b_2 \cdots b_m) $ be two partitions of $n.$ 
If the first nonvanishing $ a_i - b_i $ is positive, then we write $ A > B. $ 
This determines an ordering on the set of partitions of $n$. 

   To specify a vector in $ V^{\cal I}_n, $ we decompose $n$ into a sum of four 
non-negative integers:
\[
  n = a + b + c + d.   
\]
Let $ A, B, C $ and $ D $ be partitions of $ a, b, c $ and $ d, $ respectively. 
A vector in $ V^{\cal I}_n $ is labelled as follows:
\beq
  P_{-A}^1 P_{-B}^2 L_{-C} J_{-D} \ket{0},  \label{vecV}
\eeq
where $ P_{-A}^1 = P_{-a_1}^1 P_{-a_2}^1 \cdots P_{-a_{\ell(A)}}^1 $ and so on. 
If the decomposition of $n$ contains zero, then the corresponding partition is empty set $\phi$ and 
the corresponding generators do not appear in (\ref{vecV}). For instance, if $ a = 0, $ then 
(\ref{vecV}) becomes $ P_{-B}^2 L_{-C} J_{-D} \ket{0}. $   
For a given decomposition $(a,b,c,d) $ there are $ p(a) p(b) p(c) p(d) $ vectors 
of the form (\ref{vecV}). Here we set $ p(0) = 1 $ as usual. It follows that
\[
  \mbox{dim} V_n^{\cal I} = \sum_{(a,b,c,d)} p(a) p(b) p(c) p(d).
\]
The values of $ \mbox{dim}V^{\cal I}_n $  for some small $n$ are indecated below: 
\[
  \begin{array}{c|cccccc}
    n & \quad 0 & \quad 1 &\quad 2 &\quad 3 &\quad 4 & \quad 5\\ \hline
    \mbox{dim}V^{\cal I}_n\quad & \quad 1 &\quad 4 &\quad 14 &\quad 40 &\quad 105 & \quad 252
  \end{array}
\]
One sees that the dimension of $ V_n^{\cal I} $ increase very rapidly as a function of $n.$ 

  In the calculation of next subsection we use another notation more frequently, 
since it is more convenient to make $P^1, P^2 $ a pair and 
$ L, J$ another pair. Suppose that $n$ is decomposed into a pair of non-negative 
integers $ (a,b),$ \textit{i.e.,} $ n = a + b. $ 
We then choose  partitions $ A = (a_1 \cdots a_q) $ of $a $ and 
 $ B = (b_1 \cdots b_m) $ of $b.$ 
For a given quartet $ (a,b,A,B) $ we define a subspace $\Rep$ of $ V_n^{\cal I} $ 
such that  $ \displaystyle V_n^{\cal I} = \bigoplus_{(a,b,A,B)} \Rep. $  
The basis of $ \Rep $ is determined as follows: 
First we produce two partitions $ A_1, A_2 $ form the partition $A = (a_1 \cdots a_q)$. 
The partition $ A_1 $ is a sequence of $s \;(0 \leq s \leq q)$ integers selected from $ a_1, a_2, \cdots, a_q, $ 
and the remaining $ q -s $ integers defines the partition $ A_2. $ Namely,
\bea
  & & A_1=(a_{\sigma_1} a_{\sigma_2} \cdots a_{\sigma_{s}}), 
      \qquad 
      A_2 = (a_{\rho_1} a_{\rho_2} \cdots a_{\rho_{q-s}}),
   \nn \\
  & & \deg A_1 + \deg A_2 = \deg A = a, \quad \deg A_k \leq \deg A, \ (k=1,2)
     \label{subpart}
\eea
In a similar way we produce two partitions $ B_1, B_2$ from the partition $ B= (b_1 \cdots b_m): $
\bea
  & & B_1 = (b_{\lambda_1}  b_{\lambda_2}  \cdots  b_{\lambda_t}), 
    \quad 
      B_2 = (b_{\nu_1}  b_{\nu_2}  \cdots  b_{\nu_{m-t}}),
    \quad (0 \leq t \leq m) 
    \nn \\
   & & \deg B_1 + \deg B_2 = \deg B = b, \quad \deg B_k \leq \deg B, \ (k= 1,2)
   \label{subpart2}
\eea
We associate the vector $ P_{-A_1}^1 P_{-A_2}^2 L_{-B_1} J_{-B_2} \ket{0} \in \Rep $ with 
each quartet of partitions $ (A_1,A_2,B_1,B_2). $ Namely, the basis of $ \Rep $ is labelled by 
the partitions $ (A_1,A_2,B_1,B_2). $ Hence $ \dim \Rep $ is equal to the number of all 
possible partitions $ (A_1,A_2,B_1,B_2) $  for the given $ (a,b,A,B). $ 
For short we denote a vector in $ \Rep $ by
\[
 \ket{ (P^1P^2)_{-A} (LJ)_{-B}  } \quad\mbox{or}\quad  (P^1P^2)_{-A} (LJ)_{-B}  \ket{0}. 
\]

  For illustration  the  vectors belonging to $ \Rep $ for $ n= 1, 2 $ are listed below. 
The vectors and $ \Rep $ are written in \textit{horizontal order} (see  Definition \ref{def:horizon}).
\[
  {\cal V}^{1,0}_{(1) \phi}:\ P^1_{-1} \ket{0},\ P^2_{-1} \ket{0}, \qquad
  {\cal V}^{0,1}_{\phi (1)}:\  L_{-1} \ket{0},\ J_{-1} \ket{0}.
\]
\bea
   & & 
   {\cal V}^{2,0}_{(1^2) \phi}:\ (P^1_{-1})^2 \ket{0},\ P^1_{-1} P^2_{-1} \ket{0},\ (P^2_{-1})^2 \ket{0}, 
   \qquad
   {\cal V}^{2,0}_{(2) \phi}:\  P^1_{-2}\ket{0},\ P^2_{-2} \ket{0},
   \nn \\
   & & {\cal V}^{1,1}_{(1)(1)}:\  P^1_{-1} L_{-1} \ket{0},\  P^1_{-1} J_{-1} \ket{0},\ 
              P^2_{-1} L_{-1} \ket{0},\  P^2_{-1} J_{-1}\ket{0},
   \nn \\
   & &
   {\cal V}^{0,2}_{\phi (2)}:\ L_{-2} \ket{0},\ J_{-2} \ket{0}, 
   \qquad
   {\cal V}^{0,2}_{\phi (1^2)}:\ (L_{-1})^2 \ket{0},\  L_{-1} J_{-1} \ket{0},\  (J_{-1})^2 \ket{0}.
   \nn 
\eea

%
\subsection{Kac determinant formula}

  The reducibility of  $ V^{\cal I} $ may be investigated by the Kac determinant. 
The Kac determinant is defined as usual~\cite{KaRa}. 
Let $ \ket{i} (i = 1, \cdots \mbox{dim} V_n^{\cal I}) $ be a basis of $ V_n^{\cal I}, $ then 
the Kac determinant at level (degree) $n$ is given by
\[
  \Delta_n = \det (\; \bracket{i}{j} \;).
\]

  The essential idea for calculating $ \Delta_n $ for arbitrary $n$ is to define 
two different orderings for the basis of $ V_n^{\cal I}. $ By the orderings $ \Delta_n $ 
is equal (up to sign) to the determinant of a matrix of row echelon form. 
This will be achieved by  the following lemmas and definitions:
\begin{lemma}
\label{lemm:vanish}
  Let $ {\cal V}^{a,b}_{AB},\; {\cal V}^{c,d}_{CD} \subset V_n^{\cal I}. $ Then 
\beq
   \bracket{ (P^1P^2)_{-A} (LJ)_{-B}}{ (P^1P^2)_{-C} (LJ)_{-D} } = 0, \label{InProRel1}
\eeq
if one of the followings is true:
\begin{enumerate}
 \renewcommand{\labelenumi}{\roman{enumi})}
  \item $ a > d $ ( so that $b < c$)
  \item $ a = d $ ( so that $ b = c $) and  $ A < D $
  \item $ a = d $ ( so that $ b = c $)  and  $ B > C $
\end{enumerate}
\end{lemma}
\begin{proof}
 By the definition of the inner product and the commutativity of $ P_n^i, $ the inner product yields
\[
  \bracket{ (P^1P^2)_{-A} (LJ)_{-B}}{ (P^1P^2)_{-C} (LJ)_{-D} }
  = 
  \bracket{\omega(  (P^1P^2)_{-C} ) (LJ)_{-B} }{ \, \omega( (P^1P^2)_{-A} ) (LJ)_{-D} }
\]
We show that $ \ket{ \omega( (P^1P^2)_{-A} ) (LJ)_{-D} } = 0 $ if the condition i) or ii) is true. 

\noindent
i) $ a > d. $ 
We move $ (LJ)_{-D} $ to the left of $ \omega( (P^1P^2)_{-A} ): $
\bea
  & & \ket{ \omega( (P^1P^2)_{-A} ) (LJ)_{-D} }
  \nn \\
  & & \qquad 
   = (LJ)_{-D}\, \omega( (P^1P^2)_{-A} ) \ket{0} + \sum_{E,F,G} f_{E,F,G} (LJ)_{-E} (P^1P^2)_{-F} (P^1P^2)_{G} \ket{0},
  \label{L1Proof1}
\eea
where $ (P^1P^2)_{G} \in U(\calg^+) $ and $ f_{E,F,G} $ is a numerical coefficient. 
The associated identity for the degree of partition 
is 
\[
  a -d = -\deg E - \deg F + \deg G.
\]
The first term of (\ref{L1Proof1}) 
vanishes since $ \omega( (P^1P^2)_{-A} ) \in U(\calg^+). $ 
The condition $ a > d $ means that $ \deg G > 0,  $ i.e., $ (P^1P^2)_{G} $ always exists 
in the second term of (\ref{L1Proof1}). This factor  annihilates the highest weight vector $ \ket{0} $ 
so that the second term vanishes, too. Hence the right hand side of (\ref{L1Proof1}) is always zero.

\noindent
ii) $ a = d $ and $ A < D. $ 
Let $ D= (d_{\mu_1} d_{\mu_2}\cdots). $  Then one may write $ (LJ)_{-D} = L_{-d_{\mu_1}} (LJ)_{-D^{(1)}} $ where 
$ D^{(1)} = (d_{\mu_2} d_{\mu_3} \cdots) $ is a partition of $ d - d_{\mu_1}. $  
Because of the Abelian nature of $ \langle \; P^i_n \; \rangle $ one can write 
$ \omega( (P^1P^2)_{-A} ) = (P^1 P^2)_{A}. $ 
It follows that
\bea
 & & \ket{ \omega( (P^1P^2)_{-A} ) (LJ)_{-D} } = [\,(P^1P^2)_A, (LJ)_{-D} \,] \ket{0}
 \nn \\
 & & \qquad  
  = [\,(P^1P^2)_A, L_{-d_{\mu_1}}] (LJ)_{-D^{(1)}} \ket{0} 
  + 
  L_{-d_{\mu_1}}  [\,(P^1P^2)_A, (LJ)_{-D^{(1)}} ] \ket{0}
  \nn \\
 & & \qquad 
 = \sum_{A^{(1)},B^{(1)}} (P^1P^2)_{-B^{(1)}} (P^1P^2)_{A^{(1)}} (LJ)_{-D^{(1)}} \ket{0} 
 + L_{-d_{\mu_1}} \omega( (P^1P^2)_{-A} )   (LJ)_{-D^{(1)}} \ket{0},
 \nn \\
 & & 
  \label{L1Proof2}
\eea
where the numerical coefficients appearing in the summation part are omitted for the sake of 
simplicity. 
The second term of (\ref{L1Proof2}) has no contribution because of i) and $ a > d-d_{\mu_1}. $ 
In the first term (summation part) we have the relation
\[
  \deg A^{(1)} = \deg D^{(1)} + \deg B^{(1)}. 
\]
If $ \deg B^{(1)} > 0 $ then $ \deg A^{(1)} > \deg D^{(1)}. $  We use i) again and see that the terms 
with $ \deg B^{(1)} > 0 $ has no contribution to the summation. Thus only the terms with $ \deg B^{(1)} = 0 $ 
remains:
\[
  \ket{ \omega( (P^1P^2)_{-A} ) (LJ)_{-D} } = 
  \sum_{A^{(1)}} \ket{ (P^1P^2)_{A^{(1)}} (LJ)_{-D^{(1)}} },
\]
where 
\[
   \deg A^{(1)} = \deg D^{(1)} < \deg A, \qquad  A^{(1)} < D^{(1)}.
\]
The second relation is due to $ A < D. $ 

  One can repeat the same argument for $ \ket{ (P^1P^2)_{A^{(1)}} (LJ)_{-D^{(1)}} } $ again and again. 
At every step we have 
\[
  \ket{ \omega( (P^1P^2)_{-A} ) (LJ)_{-D} } = 
  \sum_{A^{(k)}} \ket{ (P^1P^2)_{A^{(k)}} (LJ)_{-D^{(k)}} },
\]
where 
\[
   \deg A^{(k)} = \deg D^{(k)} < \deg A^{(k-1)}, \qquad  A^{(k)} < D^{(k)}.
\]
However one can not repeat this until $ \deg A^{(k)} = 0 $ since this is contradict with $ A^{(k)} < D^{(k)}. $
This means that at certain step $ \deg B^{(k)} \neq 0 $ for all terms in the summation.  
Thus $ \ket{ \omega( (P^1P^2)_{-A} ) (LJ)_{-D} } = 0. $ 

  We have shown that (\ref{InProRel1}) is true under the condition i) or ii). 
If the condition iii) is true, then one can show that $ \ket{\omega(  (P^1P^2)_{-C} ) (LJ)_{-B} } = 0 $ 
by the same method as the cases i) and ii). 
\end{proof}

  Now we introduce two different orderings of the basis of $ V_n^{\cal I}. $ 
Essentially, they are  the orderings of the set of subspaces $ {\cal V}^{a,b}_{AB} $ 
and the ordering of vectors in each subspace is 
not essential for the calculation of $\Delta_n. $ 
\begin{definition}
\label{def:horizon}
  By the horizontal ordering we mean the following arrangement of the vectors in $ V_n^{\cal I}. $
  \begin{enumerate}
     \item we put $ {\cal V}^{a,b}_{AB} $ from left to right in decreasing order of $a $
     \item $ {\cal V}^{a,b}_{AB} $ having the same value of $a$ are rearranged  in 
     increasing order of the partition $A$
     \item $ {\cal V}^{a,b}_{AB} $ having the same value of $a$ and the same partition $A$ are 
     rearranged in decreasing order of partition $ B $
     \item vectors in each $ {\cal V}^{a,b}_{AB} $ are arranged in lexicographic 
     order with respect to $ P^1 < P^2 <  L <  J. $ The same type of generators are arranged in 
     increasing order of their indices
  \end{enumerate}
\end{definition}
The horizontal ordering for $ n = 1, 2 $ is given at the end of previous subsection. 
Further example for $ n = 3 $ is given below in terms of the subspace $ {\cal V}^{a,b}_{AB}. $
\bea
 & & {\cal V}^{3,0}_{ (1^3) \phi}, \quad {\cal V}^{3,0}_{(21) \phi}, \quad {\cal V}^{3,0}_{(3)\phi}, 
 \quad 
 {\cal V}^{2,1}_{ (1^2)(1)}, \quad {\cal V}^{2,1}_{ (2)(1)},
 \nn \\
 & & {\cal V}^{1,2}_{ (1) (2)}, \quad {\cal V}^{1,2}_{ (1) (1^2)}, \quad {\cal V}^{0,3}_{ \phi (3)}, \quad 
 {\cal V}^{0,3}_{ \phi (21)}, \quad {\cal V}^{0,3}_{ \phi (1^3)}. 
 \nn
\eea
The vectors in $ {\cal V}^{3,0}_{  (1^3) \phi } $ are arranged as
\[
  (P^1_{-1})^3 \ket{0},\quad (P^1_{-1})^2 P^2_{-1} \ket{0}, \quad 
  P^1_{-1} (P^2_{-1})^2 \ket{0}, \quad (P^2_{-1})^3 \ket{0},
\]
and the vectors in $ {\cal V}^{1,2}_{(1)(2)} $ as
\[
   P^1_{-1} L_{-2} \ket{0}, \quad P^1_{-1}  J_{-2} \ket{0}, \quad
   P^2_{-1} L_{-2} \ket{0}, \quad P^2_{-1} J_{-2} \ket{0}. 
\]
\begin{definition}
\label{def:verti}
Arrange the basis vector of $ V_n^{\cal I} $ in horizontal order. 
We replace each vector $ P^1_{-A} P^2_{-B} L_{-C} J_{-D} \ket{0} $ in the horizontal ordering with 
$ P^1_{-D} P^2_{-C} L_{-B} J_{-A} \ket{0}. $  This procedure gives a new arrangement of the basis 
of $ V_n^{\cal I} $ and we refer to this arrangement as the vertical ordering. 
\end{definition}
By definition the subspace $ {\cal V}^{a,b}_{AB} $  in the horizontal ordering is 
replaced with $ {\cal V}^{b,a}_{BA} $ in the vertical ordering. 

  Now let $ \ket{H_j} \ (i=1, \cdots, \dim V_n^{\cal I} ) $ be the basis of $ V_n^{\cal I} $ in horizontal ordering and 
 $ \ket{V_j} $ be the same basis in vertical ordering. We consider the matrix 
 $ M_n = (\; \bracket{V_i}{H_j} \; ). $ It is obvious that 
\[
    \Delta_n = \det M_n \quad (\mbox{up to sign}).
\] 
It is clear from Definition~\ref{def:horizon}, \ref{def:verti} and Lemma~\ref{lemm:vanish} that $ M_n $ is a 
matrix of row echelon form and its block diagonal parts is the matrices whose entries 
are product of vectors between $ {\cal V}^{a,b}_{AB} $ and $ {\cal V}^{b,a}_{BA}. $ 
We denote the matrices sitting in the block diagonal of $M_n$ by 
\[
  {\cal M}(abAB) = (\; \bracket{ (P^1P^2)_{-A} (LJ)_{-B} }{ (P^1P^2)_{-B} (LJ)_{-A} } \; ).
\]
As an example we give $n=2$ matrix:
{\renewcommand{\arraystretch}{2.2}
\[
  M_2 = 
  \begin{array}{c}
     \\ 0\phi(1^2)  \\ 0\phi(2)  \\ 1(1)(1) \\ 2(2)\phi \\ 2(1^2)\phi
  \end{array}
  \; 
  \begin{array}{ccccc}
   2(1^2)\phi & 2(2)\phi & 1(1)(1) & 0\phi(2) & 0\phi(1^2) \\
  \multicolumn{5}{c}{
  \left(
    \begin{array}{ccccc}
      \multicolumn{1}{c|}{\cellcolor[gray]{.8} {\  {\cal M}\ \,}} & & & & \\
      \cline{1-2}
      & \multicolumn{1}{|c|}{\cellcolor[gray]{.8} {\ {\cal M}\ }} & & \mbox{\Huge{$\ast$}} & \\
      \cline{2-3}
      & & \multicolumn{1}{|c|}{\cellcolor[gray]{.8} {\ {\cal M}\ }} & & \\
      \cline{3-4}
      & \mbox{\Huge{$0$}} & & \multicolumn{1}{|c|}{\cellcolor[gray]{.8} {\ \,{\cal M}\ }} & \\
      \cline{4-5} 
      & & & & \multicolumn{1}{|c}{\cellcolor[gray]{.8} {\ \,{\cal M}\ }} 
    \end{array}
  \right)
  }
  \end{array},
\]
}
where the rows and columns are labelled by the triple $(a, A, B). $ 
 
  Summarizing the results so far, $ \Delta_n $ equals, up to sign, 
to the product of determinant of the matrices sitting in diagonal 
parts of $ M_n:$ 
\[
  \Delta_n = \prod_{a,b} \prod_{A,B} \det {\cal M}(abAB),
\]  
where the pair $ (a,b) $ runs all possible decomposition of $n$ into two non-negative integers and 
for a fixed $ (a,b) $ the pair $ (A,B) $ runs  all possible partitions of $ a $ and $b.$ 
Namely, the calculation of $ \Delta_n $ has been reduced to the calculation of $ \det {\cal M}(abAB). $  
The computation of $ \det {\cal M}(abAB) $  is further simplified by the next lemma:
\begin{lemma}
\label{lemm:prod}
\bea
   & & \bracket{ (P^1P^2)_{-A} (LJ)_{-B} }{ (P^1P^2)_{-B} (LJ)_{-A} }
   \nn \\
   & & \hspace{2cm}
   =
   \bracket{ (P^1P^2)_{-A} }{ (LJ)_{-A} } 
   \bracket{ (LJ)_{-B} }{ (P^1P^2)_{-B} }.
   \label{InProRel2}
\eea
\end{lemma}
\begin{proof}
  Proof is similar to Lemma~\ref{lemm:vanish}. The LHS of (\ref{InProRel2}) yields
  \[
     \text{LHS} =  \bracket{ \omega( (P^1P^2)_{-B} ) (LJ)_{-B} }{ \, \omega( (P^1P^2)_{-A} ) (LJ)_{-A} }.
  \]
  We repeat the same procedure as the proof of Lemma~\ref{lemm:vanish} ii). 
  Let $ (LJ)_{-A} = L_{-d_{\mu_1}} (LJ)_{-A^{(1)}} $ then one may show that
  \[
    \ket{ \omega( (P^1P^2)_{-A} ) (LJ)_{-A} } = \sum_{\bar{A}^{(1)}} f_1(\bar{A}^{(1)})
    \ket{ (P^1P^2)_{\bar{A}^{(1)}} (LJ)_{-A^{(1)}} },
  \]
  where $ \bar{A}^{(1)} $ is a partition of $ a - d_{\mu_1} $ and $ f_1(\bar{A}^{(1)}) $ is a numerical coefficient. 
  We have the relation for the partitions
  \[
     \deg \bar{A}^{(1)} = \deg A^{(1)} < \deg A, \qquad \bar{A}^{(1)} \leq A^{(1)}.
  \]
  One may repeat this again and again, then come to the equation
  \[
    \ket{ \omega( (P^1P^2)_{-A} ) (LJ)_{-A} } = \sum_{\bar{A}^{(k)}} f_k(\bar{A}^{(k)})
    \ket{ (P^1P^2)_{\bar{A}^{(k)}} (LJ)_{-A^{(k)}} },
  \]
  where
  \[
     \deg \bar{A}^{(k)} = \deg A^{(k)} < \deg A^{(k-1)}, \qquad \bar{A}^{(k)} \leq A^{(k)}.  
  \]
  Since there exits the partition $ \bar{A}^{(k)} $ equal to $ A^{(k)} $ at any step, one may repeat this 
  until $ \deg A^{(N)} = 0. $ Then we have
  \[
     \ket{ \omega( (P^1P^2)_{-A} ) (LJ)_{-A} } = f_N(A^{(N)}) \ket{0} = \bracket{ (P^1P^2)_{-A} }{ (LJ)_{-A} } \ket{0}.     
  \]
  Similarly one may prove
  \[
    \bra{ \omega( (P^1P^2)_{-B} ) (LJ)_{-B} } = \bracket{ (LJ)_{-B} }{ (P^1P^2)_{-B} }  \bra{0}.
  \]
  This proves (\ref{InProRel2}). 
\end{proof}

It follows from Lemma~\ref{lemm:prod} that $ {\cal M}(abAB) $ is a direct product of 
two matrices:
\[
  {\cal M}(abAB) = {\cal M}(A) \otimes \tilde{\cal M}(B), 
\]
where 
\[
  {\cal M}(A) = ( \;    \bracket{ (P^1P^2)_{-A} }{ (LJ)_{-A} }  \; ), 
  \quad
  \tilde{\cal M}(B) = ( \; \bracket{ (LJ)_{-B} }{ (P^1P^2)_{-B} } \; ).
\]
Let $ s(A) $ be the size of matrix $ {\cal M}(A) $ which equals to the size of $ \tilde{\cal M}(A) $ and equals 
to the number of vectors denoted by $ \ket{ (LJ)_{-A} }$  for fixed  $ a $ and $A$.  Then
\[
  \det {\cal M}(abAB) = ( \det {\cal M}(A) )^{s(B)} (\det \tilde{\cal M}(B) )^{s(A)}.
\]
We note that if $ A = \phi $ then $ \ket{(LJ)_A} = \ket{0} $ so that $ s(\phi) = 1. $ 
In this way, calculation of $ \Delta_n $ is finally reduced to calculation of 
$ \det {\cal M}(A)  $ and $  \det \tilde{\cal M}(A). $  
\begin{lemma}
\label{lemm:detM}
\[
  \det {\cal M}(A) = \det \tilde{\cal M}(A) = \lambda(a,A,\ell) (\rho_1^2 + \rho_2^2)^{\half s(A) \ell(A)},
\]
where the equality is up to sign. 
The overall factor $ \lambda(a,A,\ell) $ depends only on $ a, $ its partition $A$ and the spin parameter $ \ell. $ 
\end{lemma}
\begin{proof}
  We prove by induction with respect to $ \ell(A). $
  
\noindent
i) If $ \ell(A) = 1, \; i.e., \; A = (a), $ 
then $ s(A) = 2 $ and the possible partitions $ A_1, A_2 $ are 
$ (A_1, A_2) = ((a)\phi),\, (\phi(a)). $ 
The matrix  $ {\cal M}(A) $ and  $ \det{\cal M}(A) $ are calculated as follows
\bea
 & & 
  {\cal M}(A) = 
  \begin{pmatrix}
     \bracket{P^1_{-a}}{L_{-a}} & \bracket{P^1_{-a}}{J_{-a}} \cr
     \bracket{P^2_{-a}}{L_{-a}} & \bracket{P^2_{-a}}{J_{-a}}
  \end{pmatrix}
  = 
  \begin{pmatrix}
    (\ell+1)a \rho_1 & \rho_2 \cr
    (\ell+1)a \rho_2 & -\rho_1
  \end{pmatrix},
 \nn \\
 & & 
   \det{\cal M}(A) = -(\ell+1)a (\rho_1^2 + \rho_2^2).
  \label{l=1}
\eea
Hence the lemma is true for this case. 

\noindent
ii) Suppose that the lemma is true for any partition $ A = (a_1 a_2 \cdots a_{\ell(A)} ) $ of length $ \ell(A). $ 
Consider a partition $ A' $ which is obtained by adding one more positive integer $ \alpha $ to the partition $A: $ 
\[
  A' = (a_1 a_2 \cdots \alpha \cdots a_{\ell(A)}), \qquad \ell(A') = \ell(A) + 1.
\]
This partition $ A' $ is also obtained by adding  $ a_1 $ to the partition $ B = (a_2 \cdots \alpha \cdots a_{\ell(A)}). $ 
By the assumption of the induction, the lemma is true for the partition $ B. $ 
Thus it is enough to consider the $ A' $ of the form $ A' = (\alpha a_1 a_2 \cdots a_{\ell(A)} ) $ with 
$ \alpha \geq a_1. $ 

\noindent
iii) Suppose that $ \alpha > a_1. $ 
Let us recall the partitions defined in (\ref{subpart}). 
There exist two pairs of $(A_1',A_2') $ associated 
with one given pair $ (A_1, A_2),$ since $ A_1'$ or $ A_2'$ must contain $ \alpha. $ 
It follows that $ s(A') = 2 s(A). $ 
If $ A_1' $ contains $ \alpha,$ then we have
\[
  \ket{ (P^1P^2)_{-A'} }= P^1_{-\alpha} \ket{ (P^1P^2)_{-A} },
  \qquad 
  \ket{(LJ)_{-A'} } = L_{-\alpha} \ket{(LJ)_{-A}}.
\]
If $ A_2' $ contains $ \alpha, $ then we move $ J_{-\alpha} $ and $ P^2_{-\alpha} $ to the 
left most position:
\bea
  \ket{ (P^1P^2)_{-A'} }= P^2_{-\alpha} \ket{ (P^1P^2)_{-A} }, 
 & &
  \ket{(LJ)_{-A'} } = J_{-\alpha} \ket{(LJ)_{-A}} + \sum_{B} \ket{(LJ)_{-B}},
  \nn
\eea
where the summation runs over some partitions $ B $ satisfying
\beq
  \deg B = \deg A' = a + \alpha, \qquad B > A'. \label{LMcond}
\eeq
By Lemma~\ref{lemm:vanish}, the summation parts of $  \ket{(LJ)_{-A'} } $ do not contribute 
to the matrix elements of $ {\cal M}(A'). $ In short hand notation the matrix $ {\cal M}(A') $ 
may be written  as follows:
\beq
 {\cal M}(A') = 
 \begin{pmatrix}
      \bracket{ P^1_{-\alpha}(P^1P^2)_{-A} }{ L_{-\alpha} (LJ)_{-A}   } 
      &  
      \bracket{ P^1_{-\alpha}(P^1P^2)_{-A} }{ J_{-\alpha} (LJ)_{-A}   }
      \\
      \bracket{ P^2_{-\alpha}(P^1P^2)_{-A} }{ L_{-\alpha} (LJ)_{-A}   } 
      &  
      \bracket{ P^2_{-\alpha}(P^1P^2)_{-A} }{ J_{-\alpha} (LJ)_{-A}   }      
 \end{pmatrix}.
 \label{SymbolicFormofM}
\eeq
The matrix entries are calculated in the following way:
\bea
  & & \bracket{ P^k_{-\alpha}(P^1P^2)_{-A} }{ L_{-\alpha} (LJ)_{-A}   }
  = \bracket{ L_{\alpha} P^k_{-\alpha}(P^1P^2)_{-A} }{  (LJ)_{-A}   }
  \nn \\
  & & \quad 
  = (\ell+1) \alpha \rho_k \bracket{ (P^1P^2)_{-A} }{  (LJ)_{-A}   }
  + \bracket{ P^k_{-\alpha} L_{\alpha} (P^1P^2)_{-A} }{  (LJ)_{-A}   },
  \label{L3:Proof1}
\eea
where $ k = 1, 2 $ and the second term in the last equation vanishes because $ \alpha > a_j $ 
for all $ j= 1, 2, \dots, \ell(A).$  Similarly one has
\beq
 \bracket{ P^k_{-\alpha}(P^1P^2)_{-A} }{ J_{-\alpha} (LJ)_{-A}   }
 = 
 \sum_j \epsilon_{kj} \rho_j \bracket{ (P^1P^2)_{-A} }{  (LJ)_{-A}   }.
 \label{L3:Proof2}
\eeq
It follows that
\bea
 & & 
   {\cal M}(A') = 
     \begin{pmatrix}
        (\ell+1) \alpha \rho_1  \bracket{ (P^1P^2)_{-A} }{  (LJ)_{-A}   }
        &
        \rho_2 \bracket{ (P^1P^2)_{-A} }{  (LJ)_{-A}   }
        \\
        (\ell+1) \alpha \rho_2 \bracket{ (P^1P^2)_{-A} }{  (LJ)_{-A}   }
        &
        -\rho_1 \bracket{ (P^1P^2)_{-A} }{  (LJ)_{-A}   }
     \end{pmatrix}
  \nn \\
  & & 
   =
   {\cal M}((\alpha)) \otimes {\cal M}(A),  \nn
\eea
where $ {\cal M}((\alpha)) $ is the $ \ell(A) = 1 $ matrix given in  (\ref{l=1}). 
It follows that
\bea
 & & 
  \det {\cal M}(A') = [\,\det {\cal M}((\alpha))\,]^{s(A)} (\det {\cal M}(A))^2  
  \nn \\
 & & \qquad 
 \sim (\rho_1^2 + \rho_2^2)^{s(a) (\ell(A)+1)} 
  = (\rho_1^2 + \rho_2^2)^{\half s(A') \ell(A')},
\eea
where the overall factor $ \lambda(a,A,\ell) $ is omitted. 
Hence the lemma is true for this case.

\noindent
iv) Suppose that $ \alpha = a_1 = a_2 = \cdots = a_m > a_{m+1}, $ that is, 
$ A = (\alpha^m a_{m+1} \cdots a_{\ell(A)}) $ and $ A'  = (\alpha^{m+1} a_{m+1} \cdots a_{\ell(A)}). $
We repeat the same computation as iii) and find no difference up to the equation (\ref{SymbolicFormofM}). 
A difference appears in the last equation of (\ref{L3:Proof1}). 
The second term of the last equation in (\ref{L3:Proof1}) does not vanish in this case. 
To calculate the contribution from  the second term, we consider the partitions $ A_1, A_2 $ for 
$ (P^1 P^2)_{-A} $ of the following form:
\beq
  A_1 = (\alpha^{m_1} a_{\sigma_1} a_{\sigma_2} \cdots), \quad 
  A_2 = (\alpha^{m_2} a_{\mu_1} a_{\mu_2} \cdots), \quad 
  m_1 + m_2 = m. 
  \label{L3:Proof3}
\eeq
Then it is not difficult to see that the equation (\ref{L3:Proof1}) yields
\bea
 & & 
  \bracket{ P^1_{-\alpha}(P^1P^2)_{-A} }{ L_{-\alpha} (LJ)_{-A}   }
    = (\ell+1) \alpha  \rho_1 (m_1+1) \bracket{ (P^1 P^2)_{-A} }{ (LJ)_{-A} } 
  \nn \\
 & & \qquad \qquad \qquad 
  + (\ell+1) \alpha\rho_2 m_2 
     \bracket{ P^1_{-(\alpha^{m_1+1} a_{\sigma_1} \cdots)} P^2_{-(\alpha^{m_2-1} a_{\mu_1}\cdots)} }{ (LJ)_{-A} },
  \label{L3:Proof4} \\
 & & 
 \bracket{ P^2_{-\alpha}(P^1P^2)_{-A} }{ L_{-\alpha} (LJ)_{-A}   }
    = (\ell+1) \alpha  \rho_2 (m_2+1) \bracket{ (P^1 P^2)_{-A} }{ (LJ)_{-A} } 
 \nn\\
 & & \qquad \qquad \qquad 
  + (\ell+1) \alpha\rho_1 m_1 
     \bracket{ P^1_{-(\alpha^{m_1-1} a_{\sigma_1} \cdots)} P^2_{-(\alpha^{m_2+1} a_{\mu_1}\cdots)} }{ (LJ)_{-A} },
  \label{L3:Proof5}
\eea
and similarly we have the following instead of (\ref{L3:Proof2})
\bea
 & & 
  \bracket{ P^1_{-\alpha}(P^1P^2)_{-A} }{ J_{-\alpha} (LJ)_{-A}   }
    =   \rho_2 (m_1+1) \bracket{ (P^1 P^2)_{-A} }{ (LJ)_{-A} } 
  \nn \\
 & & \qquad \qquad \qquad 
  -\rho_1 m_2 
     \bracket{ P^1_{-(\alpha^{m_1+1} a_{\sigma_1} \cdots)} P^2_{-(\alpha^{m_2-1} a_{\mu_1}\cdots)} }{ (LJ)_{-A} },
  \label{L3:Proof6} \\
 & & 
 \bracket{ P^2_{-\alpha}(P^1P^2)_{-A} }{ J_{-\alpha} (LJ)_{-A}   }
    = - \rho_1 (m_2+1) \bracket{ (P^1 P^2)_{-A} }{ (LJ)_{-A} } 
 \nn\\
 & & \qquad \qquad \qquad 
  + \rho_2 m_1 
     \bracket{ P^1_{-(\alpha^{m_1-1} a_{\sigma_1} \cdots)} P^2_{-(\alpha^{m_2+1} a_{\mu_1}\cdots)} }{ (LJ)_{-A} },
  \label{L3:Proof7}
\eea
if the partitions $ A_1, A_2 $ for $ (P^1P^2)_{-A} $ are given by (\ref{L3:Proof3}). 
The equations (\ref{L3:Proof4}) and (\ref{L3:Proof6}) are in the same row of the matrix $ {\cal M}(A') $ 
and so are the equations (\ref{L3:Proof5}) and (\ref{L3:Proof7}).  
One may see that the second terms of (\ref{L3:Proof4}) - (\ref{L3:Proof7}) do not contribute to 
$ \det{\cal M}(A'), $  since they are linear combinations of the other rows. 
Hence we calculate $ \det{\cal M}(A') $ in the following way:
\bea
  & & \det{\cal M}(A') 
  \nn \\
  & & \quad = \det
      \begin{pmatrix}
         (\ell+1) \alpha  \rho_1 (m_1+1) \bracket{ (P^1 P^2)_{-A} }{ (LJ)_{-A} }
         &
         \rho_2 (m_1+1) \bracket{ (P^1 P^2)_{-A} }{ (LJ)_{-A} }
         \\
         (\ell+1) \alpha  \rho_2 (m_2+1) \bracket{ (P^1 P^2)_{-A} }{ (LJ)_{-A} }
         &
         -\rho_1 (m_2+1) \bracket{ (P^1 P^2)_{-A} }{ (LJ)_{-A} }
      \end{pmatrix}
  \nn \\
  & & \quad =
    [\det {\cal M}((\alpha)) ]^{s(A)}\, \det( (m_1+1)\bracket{ (P^1 P^2)_{-A} }{ (LJ)_{-A} } ) 
  \nn \\
  & & 
  \hspace{5cm} \times \;
    \det( (m_2+1)\bracket{ (P^1 P^2)_{-A} }{ (LJ)_{-A} } )
  \nn \\
  & & \quad \sim  [\,\det {\cal M}((\alpha))\,]^{s(A)} (\det {\cal M}(A))^2  
  \sim (\rho_1^2 + \rho_2^2)^{\half s(A') \ell(A')},
  \nn 
\eea
where the numerical constants are omitted. Thus the lemma is true for this case, too. 
\end{proof}
 Now we are able to write down the explicit formula of $\Delta_n. $ 
\begin{theorem}
\label{theo:KacDet}
 Level $n$ Kac determinant is given by
\[
  \Delta_n = c_n(\ell) \prod_{a,b} \prod_{A,B} (\rho_1^2 + \rho_2^2)^{\half s(A) s(B)( \ell(A) + \ell(B))}
\]
where the pair $ (a,b) $ runs all possible decomposition of $n$ into two non-negative integers and 
the pair $ (A,B) $ runs all possible partitions of fixed $ a $ and $b.$   
The coefficient $ c_n(\ell) $ is a numerical constant depending only on $ \ell. $  
\end{theorem}
Essentially, the level $n$ Kac determinant is of the form $ \Delta_n = c_n(\ell) (\rho_1^2+\rho_2^2)^{q(n)}. $ 
We give examples of $ q(n) $ for $ n = 1, 2, 3: $
\[
  q(1) = 2, \qquad q(2) = 12, \qquad q(3) = 48.
\]
\begin{remark}
$ \Delta_n $ is independent of the central charges $ c_1, c_2 $  and $ c_3.$ 
Thus the formula of $ \Delta_n $ is common for the algebras $ \alg_{\ell} $ and $ \calg. $ 
This is also the case of 
the Schr\"odinger-Virasoro algebra in $(1+1)$ dimensional spacetime~\cite{RoUn}. 
\end{remark}
\begin{proposition}
\label{prop:Irrep}
  The Verma module $V^{\cal I}$ over $\calg$ for integer $ \ell $ is irreducible if $ \rho_1^2 + \rho_2^2 \neq 0. $ 
  On the other hand if $ \rho_1 = \rho_2 = 0, $ then $ V^{\cal I} $ is reducible. 
\end{proposition}
\begin{proof}
  The first part of the proposition is a corollary of Theorem~\ref {theo:KacDet}. 
To show the second part it is enough to prove the existence of a singular vector in $ V^{\cal I}. $ 
It is not difficult to verify that if $ \rho_1 = \rho_2 = 0 $ then 
$ (P^1_{-1} \pm i P^2_{-1}) \ket{0} \in V^{\cal I} $ is a singular vector. Thus $ V^{\cal I} $ is 
reducible.  
\end{proof}

\begin{remark}
  The statement in Proposition~\ref{prop:Irrep} is also independent of the central charges so that 
it is true for the algebras $ \alg_{\ell} $ and $ \calg. $ 
\end{remark}

\begin{remark}
Proposition~\ref{prop:Irrep} is  a sharp contrast to the finite dimensional Galilean conformal 
algebras in $(2+1)$ dimensional spacetime with an integer spin $ \ell$.  
For those finite dimensional algebras, some Verma modules for certain nonvanishing highest weights are 
reducible~\cite{Phil,AiKimSeg}.  
This is because of the central extensions of the finite dimensional algebras which make $ \langle\; P^i_r \; \rangle $ noncommutative. 
\end{remark}

%
%
%
\setcounter{equation}{0}
\section{Verma modules over $ \calg$ for half-integer $ \ell$}
\label{sec:Verma2}

 In this section the spin parameter $ \ell $ is assumed to be a positive 
half-integer. We shall show that all the Verma modules over $ \calg $ for a 
half-integer $\ell$ is reducible. The algebra $ \calg $ is spanned by 
$ L_n, J_n \ (n \in {\mathbb Z}) $ and $ P^i_r \ (r \in {\mathbb Z}+\frac{1}{2}). $ 
To define the Verma modules over $ \calg $ we introduce the triangular 
type decomposition
  \bea
    \calg &=& \calg^{-} \oplus \calg^0 \oplus \calg^+ 
    \nn \\
    &=& 
    \langle\; L_{-n}, J_{-n}, P_{-r}^i \; \rangle \; \oplus \;
    \langle\; L_{0}, J_{0} \; \rangle \; \oplus \;
    \langle\; L_{n}, J_{n}, P_{r}^i \; \rangle,
    \qquad
    n,\; r > 0 \nn
  \eea
The highest weight vector $\ket{0}$ is defined as usual:
   \bea
    & &  L_n \ket{0} = J_n \ket{0} = P_r^i \ket{0} = 0, \quad n,\; r > 0, \nn \\
    & &  L_0 \ket{0} = h \ket{0}, \quad J_0 \ket{0} = \mu \ket{0}.
      \nn
   \eea
Then the Verma modules over $ \calg $ is defined by
\[
   V^{\cal J} = U(\calg^{-}) \ket{0},
\]
where $ {\cal J} = \{\; h, \mu,  c_1, c_2, c_3 \; \}. $ 
\begin{proposition}
\label{prop:Irrep2}
  $ \ket{v_{\pm}} = (P^1_{-1/2} \pm i P^2_{-1/2}) \ket{0} \in V^{\cal J} $ are singular vectors. 
  Thus all the Verma module over $\calg $ for a half-integer $ \ell $ is reducible. 
\end{proposition}
\begin{proof}
 It is easy to verify the following:
\bea
  & & L_0 \ket{v_{\pm}} = \Bigl(h-\frac{1}{2}\Bigr) \ket{v_{\pm}}, \qquad 
      J_0 \ket{v_{\pm}} = (\mu \mp i) \ket{v_{\pm}},
   \nn \\
  & & X \ket{v_{\pm}} = 0, \quad {}^{\forall}X \in \calg^+.
   \nn
\eea
Thus $ \ket{v_{\pm}} $ are singular vectors in $ V^{\cal J}. $ All the Verma module has the singular vector so that 
$ V^{\cal J}$ is reducible. 
\end{proof}
\begin{remark}
 The results in Proposition~\ref{prop:Irrep2} is independent of the central charges. Thus they are true for 
 the algebras $ \alg_{\ell} $ and $ \calg. $ 
\end{remark}

%
%
%
\setcounter{equation}{0}
\section{Concluding remarks}
\label{sec:CR}

 We studied the central extensions and irreducibility of Verma modules for 
the infinite dimensional GCA introduced by Martelli 
and Tachikawa. We focused on the algebras defined in $(2+1)$ dimensional spacetime 
and showed that the subalgebra spanned by 
$ \langle\; P_r^i \; \rangle_{r \in {\mathbb Z} + \ell}^{i=1,2} $ 
does not have any kind of central extensions (Theorem~\ref{theo:ext}). 
This makes a sharp contrast to the finite dimensional counterparts. 
The finite dimensional GCA has the exotic central extension if $ \ell $ is an integer, 
and it has the \textit{mass} central extension if $ \ell $ is a half-integer. 
The Abelian nature of the subalgebra $ \langle\; P_r^i \; \rangle $ also causes 
some differences from the finite dimensional GCA in the representation theory. 
The results in Proposition~\ref{prop:Irrep} and \ref{prop:Irrep2} and their independence 
of the central charges are mainly due to this Abelian nature. 

  The results of present work will open a way of further study of representation theory of $ \calg. $ It is an important problem to obtain irreducible highest weight modules when the Verma module is reducible. 
Unitarity of the representation is of physical importance. Thus to find the conditions for unitary 
irreducible representations is a work to be done.  
Another interesting problems is a computation of characters which might have a relation to statistical 
systems. One may expect to establish some connections to the vertex operator algebras as in~\cite{ZhangDong,LiuGaoZhu,GaoJiangPei2,Radobolja}.

%
\section*{Acknowledgements} 
The authors are grateful to Yufeng Pei for pointing out the error in Theorem~\ref{theo:ext} of the first 
version of this manuscript. 
N.A. was supported by a grants-in-aid from JSPS (Contract No.23540154).

%
%
%

\end{document}